\documentclass[12pt,a4paper]{article}

\usepackage{fullpage}
\usepackage[utf8]{inputenc}

\usepackage{amssymb}
\usepackage{amsmath}
\usepackage{amsthm}
\usepackage[round]{natbib}
\usepackage[breaklinks=true]{hyperref}
\usepackage{url}
\usepackage{graphicx}
\usepackage{multirow}
\usepackage{setspace}
\usepackage{bm}

\usepackage{booktabs,rotating}
\usepackage{multirow}

\newcommand{\Rset}{\mathbb{R}}

\newcommand{\dd}{\mathrm{d}}
\newcommand{\Cb}{\mathbb{C}}
\newcommand{\CC}{\mathcal{C}}
\newcommand{\Db}{\mathbb{D}}
\newcommand{\Eb}{\mathbb{E}}
\newcommand{\EV}{\mathcal{EV}}
\renewcommand{\vec}{\bm}
\newcommand{\1}{\mathbf{1}}
\newcommand{\norm}[1]{\|{#1}\|_{\infty}}

\newtheorem{prop}{Proposition}
\newtheorem{lem}{Lemma}

\title{Large-sample tests of extreme-value dependence for multivariate copulas}

    \author{
      Ivan Kojadinovic\\ 
      \small{Laboratoire de mathématiques et applications, UMR CNRS 5142} \\
      \small{Université de Pau et des Pays de l'Adour} \\
      \small{B.P. 1155, 64013 Pau Cedex, France} \\
      \small{\texttt{ivan.kojadinovic@univ-pau.fr}}
      \and
      Johan Segers\\ 
      \small{Institut de statistique, biostatistique et sciences actuarielles} \\
      \small{Université catholique de Louvain} \\
      \small{Voie du Roman Pays 20, B-1348 Louvain-la-Neuve, Belgium} \\
      \small{\texttt{johan.segers@uclouvain.be}}
      \and
      Jun Yan\\ 
      \small{Department of Statistics} \\
      \small{University of Connecticut, 215 Glenbrook Rd. U-4120} \\
      \small{Storrs, CT 06269, USA} \\
      \small{\texttt{jun.yan@uconn.edu}}
    }

\date{}

\begin{document}
\maketitle    


\begin{abstract}
Starting from the characterization of extreme-value copulas based on max-stability, large-sample tests of extreme-value dependence for multivariate copulas are studied. The two key ingredients of the proposed tests are the empirical copula of the data and a multiplier technique for obtaining approximate $p$-values for the derived statistics. The asymptotic validity of the multiplier approach is established, and the finite-sample performance of a large number of candidate test statistics is studied through extensive Monte Carlo experiments for data sets of dimension two to five. In the bivariate case, the rejection rates of the best versions of the tests are compared with those of the test of  \cite{GhoKhoRiv98} recently revisited by \cite{BenGenNes09}. The proposed procedures are illustrated on bivariate financial data and trivariate geological data.

\medskip


\end{abstract}


\section{Introduction}

Let $\vec X$ be a $d$-dimensional random vector with continuous marginal cumulative distribution functions (c.d.f.s) $F_1,\dots,F_d$. It is then well-known from the work of \cite{Skl59} that the c.d.f.\ $F$ of $\vec X$ can be written in a unique way as
$$
F(\vec x) = C \{ F_1(x_1),\dots,F_d(x_d) \}, \qquad \vec x \in \Rset^d,
$$
where the function $C:[0,1]^d \to [0,1]$ is a copula and can be regarded as capturing the dependence between the components of $\vec X$. If, additionally, $C$ is {\em max-stable}, i.e., if
\begin{equation}
\label{ev}
C(\vec u) = \{C(u_1^{1/r},\dots,u_d^{1/r})\}^r, \qquad \forall \, \vec u \in [0,1]^d, \qquad \forall \, r > 0, 
\end{equation}
the function $C$ is an {\em extreme-value} copula. Such copulas arise in the limiting joint distributions of suitably normalized componentwise maxima \citep{Gal87,GudSeg10} and are the subject of increasing practical interest in finance \citep{McNFreEmb05}, insurance \citep{FreVal98} and hydrology \citep{SalDeMKotRos07}.

Given a random sample $\vec X_1,\dots,\vec X_n$ from c.d.f.\ $C \{ F_1(x_1),\dots,F_d(x_d) \}$, it is of interest in many applications to test whether the unknown copula $C$ belongs to the class of extreme-value copulas. A first solution to this problem was proposed in the bivariate case by \cite{GhoKhoRiv98} who derived a test based on the bivariate probability integral transformation. The suggested approach was recently improved by \cite{BenGenNes09} who investigated the finite-sample performance of three versions of the test.

The aim of this paper is to study tests of extreme-value dependence for multivariate copulas based on characterization~(\ref{ev}). The first key element of the proposed approach is the empirical copula of the data which is a nonparametric estimator of the true unknown copula. Starting from characterization~(\ref{ev}), the empirical copula can be used to derive natural classes of empirical processes for testing max-stability. As the distribution of these processes is unwieldy, one has to resort to a multiplier technique to compute approximate $p$-values for candidate test statistics. This is the second key element of the proposed approach and is based on the seminal work of \cite{Sca05} and \cite{RemSca09}, revisited recently in \cite{Seg11}. The outcome of this work is a general procedure for testing extreme-value dependence which, in principle, can be used in any dimension. 

The second section of the paper is devoted to recent results on the weak convergence of the empirical copula process obtained in \cite{Seg11}. A detailed and rigorous description of the proposed tests is given in Section~\ref{description}, while their implementation is discussed in Section~\ref{implementation}. In the fifth section, the results of an extensive Monte Carlo study are partially reported. They are used to provide recommendations in Section~\ref{illustrations} enabling the proposed approach to be safely used to test extreme-value dependence in data sets of dimension two to five. The test based on one of the best performing statistics is finally used to test bivariate extreme-value dependence in the well-known insurance data of \cite{FreVal98}, and trivariate extreme-value dependence in the uranium exploration data of \citet{CooJoh86}. 

The following notational conventions are adopted in the sequel. The arrow~`$\leadsto$' denotes weak convergence in the sense of Definition~1.3.3 in \cite{vanWel96}, and $\ell^\infty([0,1]^d)$ represents the space of all bounded real-valued functions on $[0,1]^d$ equipped with the uniform metric. Also, for any $\vec u \in [0,1]^d$ and any $r > 0$, we adopt the notation $\vec u^r = (u_1^r,\dots,u_d^r)$. Furthermore, the set of extreme-value copulas, i.e., copulas satisfying~(\ref{ev}), is denoted by $\EV$.

Note finally that all the tests studied in this work are implemented in the R package {\tt copula} \citep{KojYan10} available on the Comprehensive R Archive Network \citep{Rsystem}.

\section{Weak convergence of the empirical copula process}
\label{empirical_copula_process}

Let $\vec{\hat U}_i = (\hat U_{i1},\dots, \hat U_{id})$, $i \in \{1,\dots,n\}$, be pseudo-observations from the copula $C$ computed from the data by $\hat U_{ij} = R_{ij}/(n+1)$, where $R_{ij}$ is the rank of $X_{ij}$ among $X_{1j},\dots,X_{nj}$. The pseudo-observations can equivalently be rewritten as $\hat U_{ij} = n \hat F_j(X_{ij})/(n+1)$, where $\hat F_j$ is the empirical c.d.f.\ computed from $X_{1j},\dots,X_{nj}$, and where the scaling factor $n/(n+1)$ is classically introduced to avoid problems at the boundary of~$[0,1]^d$. The proposed tests are based on the empirical copula of the data \citep{Deh79,Deh81}, which is usually defined as the empirical c.d.f.\ computed from the pseudo-observations, i.e.,
$$
C_n(\vec u) = \frac{1}{n} \sum_{i=1}^n \1 ( \vec{\hat U}_i \leq \vec u ), \qquad \vec u \in [0,1]^d.
$$

For any $j \in \{1,\dots,d\}$, let $C^{[j]}(\vec u)$ be the partial derivative of $C$ with respect to its $j$th argument at $\vec u$, i.e.,
$$
C^{[j]}(\vec u) = \lim_{h \to 0 \\ \atop u_j + h \in [0,1]} \frac{C(u_1,\dots,u_{j-1},u_j + h,u_{j+1},\dots,u_d) - C(\vec u)}{h},\qquad \vec u \in [0,1]^d.
$$
It is well-known \citep[see e.g.,][Theorem 2.2.7]{Nel06} that $C^{[j]}$ exists almost everywhere on $[0,1]^d$ and that, for those $\vec u \in [0,1]^d$ for which it exists, $0 \leq C^{[j]}(\vec u) \leq 1$. If $C^{[j]}(\vec u)$ exists and is continuous on $[0,1]^d$ for all $j \in \{1,\dots,d\}$, then, from Corollary~5.3 of \citet{vanWel07} \citep[see also][]{Stu84,GanStu87,FerRadWeg04,Tsu05}, the empirical copula process $\Cb_n = \sqrt{n} ( C_n - C )$ converges weakly in $\ell^\infty([0,1]^d)$ to the tight centered Gaussian process 
\begin{equation}
\label{ec}
\Cb(\vec u) = \alpha(\vec u) - \sum_{j=1}^d C^{[j]}(\vec u) \alpha(1,\dots,1,u_j,1,\dots,1), \qquad \vec u \in [0,1]^d,
\end{equation}
where $\alpha$ is a $C$-Brownian bridge, i.e., a tight centered Gaussian process on $[0,1]^d$ with covariance function $E[\alpha(\vec u)\alpha(\vec v)] = C(\vec u \wedge \vec v) -  C(\vec u)C(\vec v)$, $\vec u, \vec v \in [0,1]^d$. Without loss of generality, we assume in the sequel that $\alpha$ has continuous sample paths.

For many copula families however, the partial derivatives $C^{[j]}$, $j \in \{1,\dots,d\}$, fail to be continuous on the whole of $[0,1]^d$. For instance, as shown in \citet{Seg11}, many popular bivariate extreme-value copulas have discontinuous partial derivatives at $(0,0)$ and $(1,1)$. To deal with such situations, \cite{Seg11} considered the following less restrictive condition:
\begin{description}
\item[($\CC$)]
for any $j \in \{1,\dots,d\}$, $C^{[j]}$ exists and is continuous on the set $V_j = \{\vec u \in [0,1]^d : 0 < u_j < 1 \}$.
\end{description}
Under Condition ($\CC$), for any $j \in \{1,\dots,d\}$, \cite{Seg11} extended the domain of $C^{[j]}$ to the whole of $[0,1]^d$ by setting
$$
C^{[j]}(\vec u) = \left\{
\begin{array}{ll}
\displaystyle{\limsup_{h \downarrow 0} \frac{C(u_1,\dots,u_{j-1},h,u_{j+1},\dots,u_d)}{h}}, & \mbox{if } \vec u \in [0,1]^d, \, u_j = 0, \\
\displaystyle{\limsup_{h \downarrow 0} \frac{C(\vec u) - C(u_1,\dots,u_{j-1},1-h,u_{j+1},\dots,u_d)}{h}}, & \mbox{if } \vec u \in [0,1]^d, \, u_j = 1,
\end{array}
\right. 
$$
which ensures that the process $\Cb$ defined in~(\ref{ec}) is well-defined on the whole of $[0,1]^d$, and showed the weak convergence of the empirical copula process $\Cb_n$ to $\Cb$ in $\ell^\infty([0,1]^d)$. Condition ($\CC$) was verified in \cite{Seg11} for many popular families including many $d$-dimensional extreme-value copulas.

\section{Description of the tests}
\label{description}

Starting from characterization~(\ref{ev}) and having at hand a nonparametric copula estimator such as $C_n$, it seems natural to base tests of the hypothesis $H_0: C \in \EV$ on processes of the form 
\begin{equation}
\label{test_process}
\Db_{r,n}(\vec u) = \sqrt{n} \left[ \{ C_n(\vec u^{1/r}) \}^r - C_n(\vec u) \right], \qquad \vec u \in [0,1]^d,
\end{equation}
where $r > 0$. Alternatively, since characterization~(\ref{ev}) can equivalently be rewritten as 
$$
\{C(\vec u)\}^r = C(\vec u^r), \qquad \forall \, \vec u \in [0,1]^d, \qquad \forall \, r > 0, 
$$
one could also consider test processes of the form
$$
\Eb_{r,n}(\vec u) = \sqrt{n} \left[ C_n(\vec u^r)  - \{C_n(\vec u)\}^r \right], \qquad \vec u \in [0,1]^d.
$$ 
For a given value of $r$, such processes can be used to test the hypothesis
$$
H_{0,r}: C(\vec u) = \{ C(\vec u^{1/r}) \}^r \quad \forall \, \vec u \in [0,1]^d.
$$
Since $H_0 = \bigcap_{r > 0} H_{0,r}$, testing $H_{0,r}$ for a fixed value of $r$ is clearly not equivalent to testing~$H_0$. It follows that tests based on $\Db_{r,n}$ or $\Eb_{r,n}$, with $r$ fixed, will only be consistent for copula alternatives for which there exists $\vec u \in [0,1]^d$ such that $C(\vec u) \neq \{ C(\vec u^{1/r}) \}^r$. 

In our Monte Carlo experiments, values of $r$ smaller than one did not lead to well-behaved tests. Besides, the processes $\Db_{r,n}$ always led to consistently more powerful tests than the processes $\Eb_{r,n}$. For the sake of brevity, we therefore only present the derivation of the tests based on $\Db_{r,n}$ with $r \geq 1$.


The following result, whose short proof is given in Appendix~\ref{proofs}, gives the asymptotic behavior of the test process~(\ref{test_process}) under $H_{0,r}$.

\begin{prop}
\label{limitH0}
Suppose that the partial derivatives of $C$ satisfy Condition ($\CC$), and let $r \geq 1$. Then, under $H_{0,r}$, 
\begin{equation}
\label{Db_r}
\Db_{r,n}(\vec u) \leadsto \Db_r(\vec u) = r \{ C(\vec u^{1/r}) \}^{r-1}  \Cb(\vec u^{1/r}) - \Cb(\vec u) 
\end{equation}
in $\ell^\infty([0,1]^d)$.
\end{prop}

Before suggesting two candidate test statistics based on $\Db_{r,n}$, let us first explain how, for large $n$, approximate independent copies of $\Db_r$ can be obtained by means of a multiplier technique initially proposed in \citet{Sca05} and \cite{RemSca09}, and recently revisited in \cite{Seg11}.

As can be seen from~(\ref{Db_r}), to obtain approximate independent copies of $\Db_r$, it is necessary to obtain approximate independent copies of $\Cb$. To estimate the unknown partial derivative $C^{[j]}$, $j \in \{1,\dots,d\}$, appearing in the expression of $\Cb$ given in~(\ref{ec}), we use the estimator defined by
\begin{multline}
\label{estimator}
C^{[j]}_n(\vec u) = \frac{1}{u_{j,n}^+ - u_{j,n}^-}  \left\{ C_n (u _1,\dots,u_{j-1},u_{j,n}^+,u_{j+1},\dots,u_d ) \right. \\ \left. \mbox{}- C_n ( u_1,\dots,u_{j-1},u_{j,n}^-,u_{j+1},\dots,u_d ) \right\}, \qquad \vec u \in [0,1]^d,
\end{multline}
where $u_{j,n}^+ = (u_j + n^{-1/2}) \wedge 1$, and $u_{j,n}^- = (u_j - n^{-1/2}) \vee  0$.

This estimator differs slightly from the one initially proposed in \cite{RemSca09}. It has the advantage of converging in probability to $C^{[j]}$ uniformly over $[0,1]^d$ if $C^{[j]}$ happens to be continuous on $[0,1]^d$ instead of only satisfying Condition ($\CC$). This point is discussed in more detail in Appendix~\ref{estimators_partial_deriv}.

Let us now introduce additional notation. Let $N$ be a large integer and let $Z_i^{(k)}$, $i \in \{1,\dots,n\}$, $k \in \{1,\dots,N\}$, be i.i.d.\ random variables with mean 0 and variance 1 satisfying $\int_0^\infty \{ \Pr(|Z_i^{(k)}| > x) \}^{1/2} \dd x < \infty$, and independent of the data $\vec X_1,\dots,\vec X_n$. For any $k \in \{1,\dots,N\}$ and any $\vec u \in [0,1]^d$, let
$$
\alpha_n^{(k)}(\vec u) = \frac{1}{\sqrt{n}} \sum_{i=1}^n Z_i^{(k)} \left\{ \1(\vec{\hat U}_i \leq \vec u) - C_n(\vec u) \right\} = \frac{1}{\sqrt{n}} \sum_{i=1}^n (Z_i^{(k)} - \bar Z^{(k)}) \1(\vec{\hat U}_i \leq \vec u),
$$
where $\bar Z^{(k)} = n^{-1} \sum_{i=1}^n Z_i^{(k)}$. Furthermore, for any $k \in \{1,\dots,N\}$ and any $\vec u \in [0,1]^d$, let
$$
\Cb_n^{(k)}(\vec u) = \alpha_n^{(k)}(\vec u) - \sum_{j=1}^d C^{[j]}_n(\vec u) \alpha_n^{(k)}(1,\dots,1,u_j,1,\dots,1),
$$
and let 
$$
\Db_{r,n}^{(k)}(\vec u) = r \{ C_n(\vec u^{1/r}) \}^{r-1} \Cb_n^{(k)}(\vec u^{1/r}) -  \Cb_n^{(k)}(\vec u).
$$

The following result, whose proof is given in Appendix~\ref{proofs}, is at the root of the proposed class of tests.

\begin{prop}
\label{multiplier}
Suppose that the partial derivatives of $C$ satisfy Condition ($\CC$), and let $r \geq 1$. Then, under $H_{0,r}$, 
$$
\left(\Db_{r,n},\Db_{r,n}^{(1)}, \dots,\Db_{r,n}^{(N)} \right) \leadsto \left( \Db_r,\Db_r^{(1)},\dots,\Db_r^{(N)} \right)
$$ 
in $\{\ell^\infty([0,1]^d)\}^{(N+1)}$, where $\Db_r^{(1)},\dots,\Db_r^{(N)}$ are independent copies of the process $\Db_r$ defined in~(\ref{Db_r}). 
\end{prop}

As candidate test statistics, we consider Cram\'er--von Mises functionals of the form
$$
S_{r,n} = \int_{[0,1]^d} \{\Db_{r,n}(\vec u)\}^2 \dd \vec u,  \qquad \mbox{and} \qquad T_{r,n} = \int_{[0,1]^d} \{\Db_{r,n}(\vec u)\}^2 \dd C_n(\vec u).
$$
Also, for any $k \in \{1,\dots,N\}$, let 
$$
S_{r,n}^{(k)} = \int_{[0,1]^d} \{\Db_{r,n}^{(k)}(\vec u)\}^2 \dd \vec u,  \qquad \mbox{and} \qquad T_{r,n}^{(k)} = \int_{[0,1]^d} \{\Db_{r,n}^{(k)}(\vec u)\}^2 \dd C_n(\vec u).
$$

The following key result is proved in Appendix~\ref{proof_Trn}.

\begin{prop}
\label{T_rn}
Suppose that the partial derivatives of $C$ satisfy Condition ($\CC$), and let $r \geq 1$. Then, under $H_{0,r}$,
$$
\left( S_{r,n},S_{r,n}^{(1)},\dots,S_{r,n}^{(N)} \right) \leadsto \left( S_r,S_r^{(1)},\dots,S_r^{(N)} \right)
$$ 
and
$$
\left( T_{r,n},T_{r,n}^{(1)},\dots,T_{r,n}^{(N)} \right) \leadsto \left( T_r,T_r^{(1)},\dots,T_r^{(N)} \right)
$$
in $[0,\infty)^{(N+1)}$,  where 
$$
S_r = \int_{[0,1]^d} \{\Db_r(\vec u)\}^2 \dd \vec u \qquad \mbox{and} \qquad T_r = \int_{[0,1]^d} \{\Db_r(\vec u)\}^2 \dd C(\vec u) 
$$
are the weak limits of $S_{r,n}$ and $T_{r,n}$, respectively, and $S_r^{(1)},\dots,S_r^{(N)}$ and $T_r^{(1)},\dots,T_r^{(N)}$ are independent copies of $S_r$ and $T_r$, respectively.
\end{prop}

The previous results suggest computing approximate $p$-values for $S_{r,n}$ and $T_{r,n}$ as
$$
\frac{1}{N} \sum_{k=1}^N \1(S_{r,n}^{(k)} \geq S_{r,n}) \qquad\mbox{and}\qquad\frac{1}{N} \sum_{k=1}^N \1(T_{r,n}^{(k)} \geq T_{r,n}),
$$
respectively. 

Notice that, when $H_{0,r}$ is not true, the processes $\Db_{r,n}^{(k)}$, $k \in \{1,\dots,N\}$, cannot be regarded anymore as approximate independent copies of $\Db_{r,n}$ under $H_{0,r}$ because $\vec X_1,\dots,\vec X_n$ is not anymore a random sample from a c.d.f.\ $C \{ F_1(x_1),\dots,F_d(x_d) \}$, where $C$ satisfies $C(\vec u) = \{C(\vec u^{1/r})\}^r$ for all $\vec u \in [0,1]^d$. This however does not affect the consistency of the procedure with respect to the hypothesis $H_{0,r}$. Indeed, the process $\Db_{r,n}$ can be decomposed as 
\begin{multline*}
\Db_{r,n}(\vec u) = \sqrt{n} \left[ \{ C_n(\vec u^{1/r}) \}^r - \{C(\vec u^{1/r}) \}^r\right] \\ - \sqrt{n} \left\{ C_n(\vec u) - C(\vec u) \right\} + \sqrt{n} \left[ \{ C(\vec u^{1/r}) \}^r - C(\vec u) \right], \qquad \vec u \in [0,1]^d.
\end{multline*}
Whether $H_{0,r}$ is false or not, provided Condition ($\CC$) is satisfied, the first and second term will jointly converge weakly to the limit established in the proof of Proposition~\ref{limitH0} (see~(\ref{limit}) in Appendix~\ref{proofs}), while, if $H_{0,r}$ is false, 
$$
\sup_{\vec u \in [0,1]^d} \sqrt{n} \left| \{ C(\vec u^{1/r}) \}^r - C(\vec u) \right| \to \infty.
$$
On the other hand, from the proof of Proposition~\ref{multiplier}, it is easy to verify that, provided Condition ($\CC$) is satisfied, and whether $H_{0,r}$ is false or not,
$$
\left(\Db_{r,n}^{(1)}, \dots,\Db_{r,n}^{(N)} \right) \leadsto \left(\Db_r^{(1)},\dots,\Db_r^{(N)} \right)
$$ 
in $\{\ell^\infty([0,1]^d)\}^N$, where $\Db_r^{(1)},\dots,\Db_r^{(N)}$ are independent copies of the process $\Db_r$ defined in~(\ref{Db_r}). It follows that the statistics $S_{r,n}$ and $T_{r,n}$, as any sensible statistic derived from the process $\Db_{r,n}$, will be consistent with respect to the hypothesis $H_{0,r}$. 

To improve the sensitivity of the proposed tests, given $p$ reals $r_1,\dots,r_p \geq 1$, we also consider tests based on statistics of the form
\begin{equation}
\label{S_T}
S_{r_1,\dots,r_p,n} = \sum_{i=1}^p S_{r_i,n} \qquad \mbox{and} \qquad T_{r_1,\dots,r_p,n} = \sum_{i=1}^p T_{r_i,n}.
\end{equation}
\section{Implementation of the tests based on \protect{$S_{r,n}$} and \protect{$T_{r,n}$}}
\label{implementation}

We first discuss the implementation of the test based on $S_{r,n}$. The implementation of the test based on $T_{r,n}$ follows immediately after a simple modification.


Given a large integer $m > 0$, we proceed by numerical approximation based on a grid of $m$ uniformly spaced points on $(0,1)^d$ denoted $\vec w_1,\dots,\vec w_m$. Then, 
$$
S_{r,n} \approx \frac{1}{m} \sum_{j=1}^m \{ \Db_{r,n}(\vec w_j) \}^2 = \frac{n}{m} \sum_{j=1}^m \left[ \{ C_n(\vec w_j^{1/r})\}^r - C_n(\vec w_j) \right]^2,
$$
and, for any $k \in \{1,\dots,N\}$,
$$
S_{r,n}^{(k)} \approx \frac{1}{m} \sum_{j=1}^m \{ \Db_{r,n}^{(k)}(\vec w_j) \}^2 = \frac{1}{m} \sum_{j=1}^m \left[ r \{ C_n(\vec w_j^{1/r}) \}^{r-1}  \Cb_n^{(k)}(\vec w_j^{1/r}) - \Cb_n^{(k)}(\vec w_j) \right]^2.
$$
To efficiently implement the test, first notice that, for any $k \in \{1,\dots,N\}$, $\Cb_n^{(k)}$ can be conveniently written as
$$
\Cb_n^{(k)}(\vec u) = \frac{1}{\sqrt{n}} \sum_{i=1}^n (Z_i^{(k)} - \bar Z^{(k)}) \left\{ \1(\vec{\hat U}_i \leq \vec u) - \sum_{l=1}^d C^{[l]}_n(\vec u)  \1(\hat U_{il} \leq u_l) \right\}, \qquad \vec u \in [0,1]^d.
$$
It follows that the $\Db_{r,n}^{(k)}(\vec w_j)$ can be expressed as
$$
\Db_{r,n}^{(k)}(\vec w_j) = r \{ C_n(\vec w_j^{1/r}) \}^{r-1}  \Cb_n^{(k)}(\vec w_j^{1/r}) - \Cb_n^{(k)}(\vec w_j) = \frac{1}{\sqrt{n}} \sum_{i=1}^n (Z_i^{(k)} - \bar Z^{(k)}) M_n(i,j)
$$
in terms of a $n \times m$ matrix $M_n$ whose $(i,j)$-element is
\begin{multline*}
M_n(i,j) = r \{ C_n(\vec w_j^{1/r}) \}^{r-1} \left\{ \1(\vec{\hat U}_i \leq \vec w_j^{1/r}) - \sum_{l=1}^d C^{[l]}_n(\vec w_j^{1/r})  \1(\hat U_{il} \leq w_{jl}^{1/r}) \right\} \\ - \left\{ 1(\vec{\hat U}_i \leq \vec w_j) - \sum_{l=1}^d C^{[l]}_n(\vec w_j)  \1(\hat U_{il} \leq w_{jl}) \right\}.
\end{multline*}
In order to carry out the test based on $S_{r,n}$, it is first necessary to compute the $n \times m$ matrix $M_n$. Then, to compute $S_{r,n}^{(k)}$, it suffices to generate $n$ i.i.d.\ random variates $Z_1^{(k)},\dots,Z_n^{(k)}$ with expectation 0, variance~1, satisfying $\int_0^\infty \{ \Pr(|Z_i^{(k)}| > x) \}^{1/2} \dd x < \infty$, and perform simple arithmetic operations involving the centered $Z_i^{(k)}$ and the columns of matrix $M_n$. In the Monte Carlo simulations to be presented in the next section, the $Z_i^{(k)}$ are taken from the standard normal distribution.

For the test based on $T_{r,n}$, clearly,
$$
T_{r,n} = \frac{1}{n} \sum_{j=1}^n  \{\Db_{r,n}(\vec{\hat U}_j)\}^2 
\qquad
\mbox{and}
\qquad
T_{r,n}^{(k)} = \frac{1}{n} \sum_{j=1}^n  \{\Db_{r,n}^{(k)}(\vec{\hat U}_j)\}^2, \qquad k \in \{1,\dots,N\}.
$$
Expressions for implementing the test then immediately follow from those given for $S_{r,n}$ and $S_{r,n}^{(k)}$: simply replace $m$ by $n$, and $\vec w_j$ by $\vec{\hat U}_j$.

\section{Finite-sample performance}
\label{finite_sample}

Extensive Monte Carlo experiments were conducted to investigate the level and power of the tests based on $S_{r,n}$ and $T_{r,n}$ for samples of size $n=100$, 200, 400 and 800. The values $2,3,\dots,9$ were considered for $r$. We also investigated the finite-sample performance of the tests based on the statistics $S_{r_1,\dots,r_p,n}$ and $T_{r_1,\dots,r_p,n}$ defined in~(\ref{S_T}). Approximate $p$-values for the latter tests can be obtained by proceeding as in the previous section. Several configurations were studied among which $(r_1,\dots,r_p) = (2,3,\dots,9)$ and $(3,4,5)$. Data sets of dimension two to five were generated, both from extreme-value and non extreme-value copulas. Given that the most frequently used bivariate exchangeable extreme-value copulas such as the Gumbel--Hougaard, Galambos, H\"usler--Reiss and Student extreme-value copulas show striking similarities for a given degree of dependence \citep[see][for a detailed discussion of this matter]{GenKojNesYan11}, only the Gumbel--Hougaard (GH) and its asymmetric version (aGH) defined using  Khoudraji's device \citep{Kho95,GenGhoRiv98,Lie08} were used in the simulations. Given an exchangeable copula $C_\theta$, Khoudraji's device defines an asymmetric version of it as
$$
C_{\theta, \boldsymbol \lambda}(\vec u) = u_1^{1-\lambda_1} \dots u_d^{1-\lambda_d} C_\theta(u_1^{\lambda_1},\dots,u_d^{\lambda_d}),
$$
for all $\vec u \in [0,1]^d$ and an arbitrary choice of $\boldsymbol \lambda = (\lambda_1,\dots,\lambda_d) \in (0,1)^d$ such that $\lambda_i \neq \lambda_j$ for some $\{i,j\} \subset \{1,\dots,d\}$. If $C_\theta$ is an extreme-value copula, then the same is true of $C_{\theta, \boldsymbol \lambda}$. Note that the asymmetric Gumbel--Hougaard (aGH) obtained from Khoudraji's device is nothing else but the asymmetric logistic model introduced in \cite{Taw88,Taw90}. In the experiments, the parameter $\theta$ of the asymmetric Gumbel--Hougaard was set to 4. In dimension two, the shape parameter vector $\boldsymbol \lambda$ was taken equal to $(\lambda_1,0.95)$, with $\lambda_1 \in \{0.2,0.4,0.6,0.8\}$, so that data generated from the corresponding copulas display various degrees of asymmetry. The corresponding values for Kendall's tau are about 0.19, 0.34, 0.48 and 0.60, respectively. In dimension three, four and five, $\boldsymbol \lambda$ was set to $(0.2,0.4,0.95)$, $(0.2,0.4,0.6,0.95)$, and $(0.2,0.4,0.6,0.8,0.95)$, respectively. 

As far as non extreme-value copulas are concerned, the Clayton (C), Frank (F), normal (N), $t$ with four degrees of freedom (t), and Plackett (P) (for dimension two only) copulas were used in the experiments. 

For each of the one-parameter exchangeable families considered in the study (GH, C, F, N, t, P), three values of the parameter were considered. These were chosen so that the bivariate margins of the copulas have a Kendall's tau of 0.25, 0.50, and 0.75, respectively.

All the tests were carried out at the 5\% significance level and empirical rejection rates were computed from 1000 random samples per scenario. For the tests based on $S_{r,n}$, the parameter $m$ defined in Section~\ref{implementation} was set to $44^2$ in dimension two, $13^3$ in dimension three, $7^4$ in dimension four, and $5^5$ in dimension five. Smaller and greater values of $m$ were also considered but this did not seem to affect the results much.

\begin{table}[tbp]\small
\begin{center}
\caption{Rejection rate (in \%) of the null hypothesis in the bivariate case as observed in 1000 random samples of size $n=100$, 200, 400 and 800 from the Gumbel--Hougaard copula (GH) and its asymmetric version (aGH) with $\theta =4$ and $\boldsymbol \lambda = (\lambda_1,0.95)$.}
\label{tab:d2level}
\begin{tabular}{ccc rrrrr r rrrrr}
\toprule
Copula & $\tau$ & $\lambda_1$ & $T_{3,n}$ & $T_{4,n}$ & $T_{5,n}$ & $T_{3,4,5,n}$ & $\hat\sigma^2_n$ & \phantom{xxx}  & $T_{3,n}$ & $T_{4,n}$ & $T_{5,n}$ & $T_{3,4,5,n}$ & $\hat\sigma^2_n$ \\
\midrule
& & & \multicolumn{5}{c}{$ n = 100$} & & \multicolumn{5}{c}{$n = 200$}\\
\cmidrule(lr){4-8}\cmidrule(lr){10-14}
GH & 0.25 &  & 5.0 & 5.1 & 5.8 & 5.4 & 5.3 &  & 4.2 & 5.1 & 5.4 & 5.0 & 5.3 \\ 
   & 0.50 &  & 3.9 & 4.3 & 4.6 & 4.3 & 4.9 &  & 3.6 & 4.1 & 4.3 & 4.0 & 5.1 \\ 
   & 0.75 &  & 3.2 & 3.2 & 3.8 & 3.5 & 5.2 &  & 2.1 & 2.5 & 2.5 & 2.3 & 4.9 \\ 
  aGH &  & 0.2 & 5.3 & 5.5 & 6.2 & 5.8 & 5.4 &  & 5.9 & 5.9 & 6.7 & 6.1 & 5.8 \\ 
   &  & 0.4 & 4.6 & 5.5 & 5.9 & 5.5 & 5.6 &  & 4.3 & 5.1 & 5.6 & 5.2 & 5.4 \\ 
   &  & 0.6 & 4.2 & 4.8 & 5.2 & 4.8 & 4.9 &  & 4.4 & 5.1 & 5.4 & 5.0 & 5.5 \\ 
   &  & 0.8 & 4.7 & 5.0 & 5.1 & 5.0 & 5.0 &  & 4.7 & 5.2 & 5.6 & 5.3 & 4.9 \\ 
[2ex]\\
& & &\multicolumn{5}{c}{$ n = 400$} & &\multicolumn{5}{c}{$n = 800$}\\
\cmidrule(lr){4-8}\cmidrule(lr){10-14}
  GH & 0.25 &  & 4.3 & 4.9 & 5.1 & 4.8 & 5.2 &  & 4.6 & 5.1 & 5.1 & 5.1 & 5.0 \\ 
   & 0.50 &  & 3.4 & 3.8 & 4.2 & 3.9 & 5.2 &  & 4.0 & 4.3 & 4.6 & 4.4 & 5.0 \\ 
   & 0.75 &  & 2.4 & 2.5 & 2.6 & 2.5 & 5.3 &  & 3.4 & 3.5 & 3.6 & 3.6 & 5.0 \\ 
  aGH &  & 0.2 & 5.4 & 5.4 & 6.0 & 5.8 & 5.9 &  & 4.9 & 5.1 & 5.4 & 5.3 & 5.4 \\ 
   &  & 0.4 & 4.3 & 5.1 & 5.5 & 5.1 & 5.1 &  & 3.3 & 4.3 & 4.5 & 4.2 & 5.4 \\ 
   &  & 0.6 & 4.6 & 4.9 & 5.0 & 4.9 & 5.1 &  & 3.6 & 4.0 & 4.7 & 4.1 & 5.3 \\ 
   &  & 0.8 & 4.5 & 4.7 & 5.2 & 4.8 & 5.0 &  & 4.1 & 4.4 & 4.6 & 4.5 & 4.9 \\  
\bottomrule
\end{tabular}
\end{center}
\end{table}

\begin{table}[tbp]\small
\addtolength{\tabcolsep}{-1pt}
\begin{center}
\caption{Rejection rate (in \%) of the null hypothesis in the bivariate case as observed in 1000 random samples of size $n=100$, 200, 400 and 800 from the Clayton (C), Frank (F), normal (N), $t$ with four degrees of freedom (t), and Plackett copula (P).}
\label{tab:d2}
\begin{tabular}{cr rrrrr c rrrrr}
\toprule
Copula & $\tau$ & $T_{3,n}$ & $T_{4,n}$ & $T_{5,n}$ & $T_{3,4,5,n}$ & $\hat \sigma_n^2$ & \phantom{xx} & $T_{3,n}$ & $T_{4,n}$ & $T_{5,n}$ & $T_{3,4,5,n}$ & $\hat \sigma_n^2$ \\
\midrule
& & \multicolumn{5}{c}{$ n = 100$} & & \multicolumn{5}{c}{$n = 200$}\\
\cmidrule(lr){3-7}\cmidrule(lr){9-13}
C & 0.25 & 74.4 & 72.2 & 72.5 & 73.8 & 79.3 &  & 93.3 & 94.2 & 94.0 & 94.6 & 97.9 \\ 
   & 0.50 & 99.1 & 98.3 & 98.2 & 98.5 & 99.5 &  & 100.0 & 100.0 & 100.0 & 100.0 & 100.0 \\ 
   & 0.75 & 99.3 & 99.7 & 99.8 & 99.9 & 100.0 &  & 100.0 & 100.0 & 100.0 & 100.0 & 100.0 \\ 
  F & 0.25 & 38.9 & 43.7 & 46.4 & 45.0 & 22.4 &  & 56.1 & 66.2 & 69.8 & 66.1 & 38.2 \\ 
   & 0.50 & 62.2 & 68.8 & 75.8 & 71.3 & 34.7 &  & 88.8 & 95.2 & 96.5 & 96.0 & 60.3 \\ 
   & 0.75 & 75.0 & 85.0 & 89.2 & 86.9 & 34.7 &  & 96.7 & 98.9 & 99.4 & 99.0 & 57.3 \\ 
  N & 0.25 & 26.9 & 25.5 & 26.2 & 26.8 & 22.3 &  & 32.5 & 38.4 & 39.5 & 38.7 & 36.2 \\ 
   & 0.50 & 27.5 & 28.8 & 30.8 & 30.8 & 36.9 &  & 44.6 & 50.2 & 52.8 & 51.0 & 63.0 \\ 
   & 0.75 & 22.0 & 24.5 & 26.9 & 26.1 & 43.9 &  & 33.9 & 46.6 & 50.7 & 46.7 & 75.9 \\ 
  P & 0.25 & 35.2 & 37.6 & 42.6 & 39.3 & 21.5 &  & 50.0 & 59.0 & 63.3 & 59.2 & 38.1 \\ 
   & 0.50 & 47.9 & 54.5 & 59.2 & 56.3 & 30.5 &  & 71.1 & 81.0 & 84.8 & 81.7 & 61.7 \\ 
   & 0.75 & 42.5 & 50.1 & 56.0 & 53.1 & 34.6 &  & 60.9 & 76.4 & 83.6 & 78.4 & 58.5 \\ 
  t & 0.25 & 15.2 & 14.0 & 14.4 & 14.8 & 15.0 &  & 17.6 & 18.7 & 18.1 & 18.4 & 26.6 \\ 
   & 0.50 & 22.8 & 22.9 & 23.3 & 23.9 & 29.4 &  & 31.1 & 33.0 & 31.4 & 33.4 & 52.7 \\ 
   & 0.75 & 19.2 & 19.2 & 20.6 & 20.7 & 39.9 &  & 26.2 & 33.5 & 33.4 & 34.2 & 69.2 \\ 
[2ex]\\
& & \multicolumn{5}{c}{$ n = 400$} & & \multicolumn{5}{c}{$n = 800$}\\
\cmidrule(lr){3-7}\cmidrule(lr){9-13}
  C & 0.25 & 99.8 & 99.8 & 100.0 & 100.0 & 100.0 &  & 100.0 & 100.0 & 100.0 & 100.0 & 100.0 \\ 
   & 0.50 & 100.0 & 100.0 & 100.0 & 100.0 & 100.0 &  & 100.0 & 100.0 & 100.0 & 100.0 & 100.0 \\ 
   & 0.75 & 100.0 & 100.0 & 100.0 & 100.0 & 100.0 &  & 100.0 & 100.0 & 100.0 & 100.0 & 100.0 \\ 
  F & 0.25 & 85.7 & 91.5 & 94.1 & 92.5 & 64.6 &  & 99.3 & 99.8 & 99.9 & 99.8 & 92.6 \\ 
   & 0.50 & 99.6 & 100.0 & 100.0 & 100.0 & 87.6 &  & 100.0 & 100.0 & 100.0 & 100.0 & 99.1 \\ 
   & 0.75 & 100.0 & 100.0 & 100.0 & 100.0 & 85.9 &  & 100.0 & 100.0 & 100.0 & 100.0 & 99.4 \\ 
  N & 0.25 & 55.6 & 60.7 & 63.7 & 61.4 & 61.7 &  & 85.9 & 89.5 & 90.7 & 90.2 & 90.0 \\ 
   & 0.50 & 72.1 & 77.9 & 80.0 & 79.0 & 90.9 &  & 96.7 & 98.2 & 98.9 & 98.3 & 99.5 \\ 
   & 0.75 & 55.8 & 70.9 & 77.3 & 72.8 & 96.3 &  & 94.2 & 97.6 & 99.0 & 98.5 & 100.0 \\ 
  P & 0.25 & 75.8 & 84.0 & 87.5 & 85.2 & 61.4 &  & 96.8 & 99.1 & 99.5 & 99.2 & 89.5 \\ 
   & 0.50 & 95.1 & 97.0 & 98.8 & 97.7 & 87.4 &  & 100.0 & 100.0 & 100.0 & 100.0 & 99.1 \\ 
   & 0.75 & 92.4 & 97.2 & 98.9 & 97.9 & 87.9 &  & 99.9 & 100.0 & 100.0 & 100.0 & 99.4 \\ 
  t & 0.25 & 28.2 & 27.5 & 25.9 & 28.0 & 44.4 &  & 52.1 & 45.6 & 44.4 & 47.6 & 74.1 \\ 
   & 0.50 & 55.3 & 57.1 & 56.2 & 58.4 & 84.2 &  & 85.6 & 86.7 & 85.5 & 87.8 & 99.0 \\ 
   & 0.75 & 49.7 & 56.3 & 56.8 & 56.8 & 93.7 &  & 87.6 & 89.6 & 90.8 & 90.9 & 100.0 \\
\bottomrule
\end{tabular}
\end{center}
\end{table}

In most scenarios involving extreme-value copulas, the tests turned out to be globally too conservative, although the agreement with the 5\% level seemed to improve as $n$ was increased. To attempt to improve the empirical levels of the tests for $n \in \{100,200\}$, we considered several asymptotically negligible ways of rescaling the empirical copula in the expression of the test process~(\ref{test_process}), while keeping the expressions of the processes $\Db_{r,n}^{(k)}$, $k \in \{1,\dots,N\}$, unchanged. Reasonably good empirical levels were obtained by replacing $C_n$ in the expression of $\Db_{r,n}$ by $n (n+0.85)^{-1} C_n$. With this asymptotically negligible modification, the best results were obtained for $r \in \{3,4,5\}$ and for the tests based on $T_{r,n}$, which consistently outperformed the tests based on $S_{r,n}$. In dimension two, the rejection rates of the tests based on $T_{3,n}$, $T_{4,n}$, $T_{5,n}$, and $T_{3,4,5,n}$ are reported in Tables~\ref{tab:d2level} and~\ref{tab:d2}. 

As can be seen from Table~\ref{tab:d2level}, the empirical levels of the selected tests are, overall, reasonably close to the 5\% nominal level for $\tau \in \{0.25,0.5\}$ and $\lambda_1 \in \{0.2,0.4,0.6,0.8\}$, which, as discussed earlier, corresponds to weak to moderate dependence. The tests remain however too conservative when $\tau=0.75$, although the empirical levels seems globally to improve as $n$ increases. An inspection of Table~\ref{tab:d2} shows that, in terms of power, the tests based on $T_{4,n}$ and $T_{5,n}$ appear globally more powerful than that based on $T_{3,n}$, although the latter sometimes outperforms the former in the case of weakly dependent data sets. As far as the test based on $T_{3,4,5,n}$ is concerned, its rejection rates are almost always greater than those of $T_{4,n}$, and sometimes greater than those of $T_{5,n}$. 

The previous tests can be compared with the test of extreme-value dependence proposed by \cite{GhoKhoRiv98} and improved by \cite{BenGenNes09}. The rejection rates of the best version of that test, based on a variance estimator denoted $\hat \sigma_n^2$, were computed using routines available in the {\tt copula} R package, and are reported in Table~\ref{tab:d2}. The test based on $\hat \sigma_n^2$ is more powerful than its competitors when data are generated from an elliptical copula, the gain in power being particularly large for the $t$ copula. The proposed tests perform better when data are generated from a Frank or a Plackett copula. For $n =100$ and the Frank copula, the rejection rates of test based on $T_{3,4,5,n}$ are approximately twice as great as those of the test based on $\hat \sigma_n^2$. From the lower right block of Table~\ref{tab:d2}, we also see that, for all tests, the optimal rejection rate is almost reached in all scenarios not involving extreme-value copulas when $n=800$.

The rejection rates of the test based on $T_{3,4,5,n}$ for data sets of dimension three, four and five are given in Table~\ref{tab:d3d4d5}. As can be seen from the first two horizontal blocks of the table, in the case of weak to moderate dependence, the test appears slightly conservative, overall, although the agreement with the 5\% level seems to improve as $n$ increases. As in dimension two, the test is the most conservative in the case of strongly dependent data and this phenomenon increases with the dimension. Notice however that, in almost all scenarios under the alternative hypothesis, the power of the test increases as $d$ increases. This might be due to the fact that every bivariate margin of a $d$-variate extreme-value copula must be max-stable. Hence, deviations from multivariate max-stability might be easier to detect as the dimension increases. Note finally that, as $n$ reaches 800, the optimal rejection rate is almost attained in all scenarios not involving extreme-value copulas.

\begin{sidewaystable}[tbp]\small
\addtolength{\tabcolsep}{-1pt}
\begin{center}
\caption{Rejection rate (in \%) of the null hypothesis for the test based on $T_{3,4,5,n}$ for $d=3$, 4 and 5 as observed in 1000 random samples of size $n=100$, 200, 400 and 800 from the Gumbel--Hougaard (GH), its asymmetric version (aGH), the Clayton (C), Frank (F), normal (N), and the $t$ copula with four degrees of freedom (t). The parameters of aGH are $\theta=4$ and $\boldsymbol \lambda = (0.2,0.4,0.95)$ in dimension three, $\boldsymbol \lambda = (0.2,0.4,0.6,0.95)$ in dimension four, and $\boldsymbol \lambda = (0.2,0.4,0.6,0.8,0.95)$ in dimension five.}
\label{tab:d3d4d5}
\begin{tabular}{cr rrrr c rrrr c rrrr}
\toprule
True & $\tau$ & \multicolumn{4}{c}{$d = 3$} & \phantom{x} &  \multicolumn{4}{c}{$d = 4$} & \phantom{x} &  \multicolumn{4}{c}{$d = 5$}\\
\cmidrule(lr){3-6}\cmidrule(lr){8-11}\cmidrule(lr){13-16}
& & 100 & 200 & 400 & 800 & &  100 & 200 & 400 & 800 & &  100 & 200 & 400 & 800\\
\midrule
GH & 0.25 & 5.0 & 4.9 & 4.8 & 4.6 &  & 4.8 & 5.0 & 4.3 & 4.9 &  & 4.2 & 4.5 & 4.6 & 4.8 \\ 
   & 0.50 & 2.8 & 3.0 & 3.4 & 4.0 &  & 2.3 & 2.8 & 3.2 & 3.6 &  & 2.2 & 2.2 & 2.7 & 3.6 \\ 
   & 0.75 & 0.9 & 1.1 & 1.6 & 2.4 &  & 0.4 & 0.5 & 0.9 & 1.8 &  & 0.2 & 0.3 & 0.6 & 1.0 \\
[1ex]\\ 
aGH &  & 5.5 & 4.4 & 4.4 & 4.8 &  & 4.2 & 3.6 & 3.9 & 4.0 &  & 3.5 & 3.4 & 3.0 & 3.4 \\ 
[1ex]\\ 
  C & 0.25 & 91.9 & 99.9 & 100.0 & 100.0 &  & 98.2 & 100.0 & 100.0 & 100.0 &  & 98.8 & 100.0 & 100.0 & 100.0 \\ 
   & 0.50 & 100.0 & 100.0 & 100.0 & 100.0 &  & 100.0 & 100.0 & 100.0 & 100.0 &  & 100.0 & 100.0 & 100.0 & 100.0 \\ 
   & 0.75 & 100.0 & 100.0 & 100.0 & 100.0 &  & 100.0 & 100.0 & 100.0 & 100.0 &  & 100.0 & 100.0 & 100.0 & 100.0 \\ 
[1ex]\\ 
  F & 0.25 & 59.0 & 86.9 & 99.4 & 100.0 &  & 62.2 & 94.6 & 99.9 & 100.0 &  & 65.6 & 96.0 & 100.0 & 100.0 \\ 
   & 0.50 & 83.3 & 98.9 & 100.0 & 100.0 &  & 88.0 & 99.8 & 100.0 & 100.0 &  & 90.1 & 99.9 & 100.0 & 100.0 \\ 
   & 0.75 & 91.0 & 100.0 & 100.0 & 100.0 &  & 89.8 & 100.0 & 100.0 & 100.0 &  & 89.3 & 99.8 & 100.0 & 100.0 \\ 
[1ex]\\ 
N & 0.25 & 35.2 & 65.2 & 91.3 & 99.1 &  & 46.9 & 76.9 & 97.5 & 100.0 &  & 52.1 & 86.3 & 99.1 & 100.0 \\ 
   & 0.50 & 39.6 & 69.9 & 94.7 & 99.9 &  & 41.1 & 76.8 & 97.4 & 100.0 &  & 45.8 & 80.7 & 98.4 & 100.0 \\ 
   & 0.75 & 20.7 & 44.8 & 85.0 & 99.6 &  & 12.9 & 42.6 & 85.7 & 99.9 &  & 10.6 & 32.7 & 84.5 & 99.9 \\ 
[1ex]\\ 
  t & 0.25 & 16.9 & 23.6 & 42.1 & 67.0 &  & 16.7 & 28.0 & 50.7 & 79.9 &  & 20.3 & 33.2 & 58.7 & 83.9 \\ 
   & 0.50 & 23.8 & 44.8 & 74.3 & 96.9 &  & 26.3 & 48.7 & 83.2 & 99.0 &  & 26.8 & 53.8 & 84.0 & 99.8 \\ 
   & 0.75 & 12.6 & 26.3 & 65.6 & 97.4 &  & 7.1 & 19.8 & 60.4 & 98.1 &  & 3.9 & 15.3 & 55.7 & 96.3 \\ 
\bottomrule
\end{tabular}
\end{center}
\end{sidewaystable}

\section{Discussion and illustrations}
\label{illustrations}

The results of the extensive Monte Carlo experiments partially reported in the previous section suggest that the test based on the statistic $T_{3,4,5,n}$ can be safely used in dimension two or greater to assess whether data arise from an extreme-value copula. The choice of the statistic $T_{3,4,5,n}$ is not claimed to be optimal as other candidate test statistics could be considered. In dimension two, the test appears more powerful than the test of \cite{BenGenNes09} based on $\hat \sigma_n^2$ in approximately half of the scenarios under the alternative hypothesis, and is outperformed in the remaining scenarios. In dimension strictly greater than two, the proposed approach is presently, to the best of our knowledge, the only available procedure for testing extreme-value dependence.  

As an illustration, we first applied the test based on $T_{3,4,5,n}$ to the bivariate indemnity payment and allocated loss adjustment expense data studied in \cite{FreVal98}. These consist of 1466 general liability claims randomly chosen from late settlement lags (among the initial 1500 claims, 34 claims for which the policy limit was reached were ignored). Many studies, including that of \cite{BenGenNes09}, have concluded that an extreme-value copula is likely to provide an adequate model of the dependence. 

Note that these data contain a non negligible number of ties. As is the case for other procedures based on the empirical copula, the presence of ties might significantly affect the tests under study since these were derived under the assumption of continuous margins. To deal somehow satisfactorily with ties, \cite{KojYan10} suggested to assign ranks at random in the case of ties when computing pseudo-observations. This was done using the R function \texttt{rank} with its argument \texttt{ties.method} set to \texttt{"random"}. The test was then carried out on the resulting pseudo-observations. With the hope that the use of randomization will result in many different configurations for the parts of the data affected by ties, the test based on the pseudo-observations computed with \texttt{ties.method = "random"} was performed 100 times with $N=1000$. The minimum, median and maximum of the obtained approximate $p$-values are 40.7\%, 45.9\%, and 50.4\%, respectively. If the pseudo-observations are computed using mid-ranks, the approximate $p$-value, based on $N=10 \, 000$ multiplier iterations, drops down to 1.7\%. As already observed in other situations, using mid-ranks seems to increase the evidence against the null hypothesis.
 
\begin{table}[tbp]
\addtolength{\tabcolsep}{-2pt}
\begin{center}
\caption{Approximate $p$-values (in \%) for the test based on $T_{3,4,5,n}$ obtained for the triples of variables \{U,Co,Li\}, \{U,Li,Ti\} and \{Ti,Li,Cs\} of the uranium data of \citet{CooJoh86}.}
\label{tab:u3}
\begin{tabular}{l rrr c r}
\toprule
 & \multicolumn{3}{c}{Random ranks for ties}& \phantom{xx} &  \\
 & Minimum & Median & Maximum & \phantom{xx} & Mid-ranks \\
\midrule
\{U, Co, Li\} & 0.0 & 0.0 & 0.1 & &0.0 \\ 
  \{U, Li, Ti\} & 0.0 & 0.1 & 0.3 & & 0.0 \\ 
  \{Ti, Li, Cs\} & 2.5 & 3.9 & 5.5 & &1.8 \\ 
\bottomrule
\end{tabular}
\end{center}
\end{table}

As a second example, we considered the uranium exploration data of \citet{CooJoh86}. The data consist of log-concentrations of seven chemical elements in 655 water samples collected near Grand Junction, Colorado: uranium (U), lithium (Li), cobalt (Co), potassium (K), cesium (Cs), scandium (Sc), and titanium (Ti). \citet{BenGenNes09} performed an extensive study of the 21 pairs of variables and suggested that the triples \{U,Co,Li\}, \{U,Li,Ti\} and \{Ti,Li,Cs\} should be investigated for trivariate extreme-value dependence once a multivariate test becomes available. Note that the number of ties in these data is greater than in the insurance data of \cite{FreVal98}. In particular, the variable Li takes only 90 different values out of 655. For that reason, as previously, we broke the ties at random and repeated the calculations 100 times with $N=1000$. Approximate $p$-values for the test based on $T_{3,4,5,n}$ are summarized in Table~\ref{tab:u3}. The first three columns give the minimum, median and maximum of the obtained $p$-values. The last column gives the $p$-values computed from the mid-ranks using $N=10\,000$. As for the insurance data, we see that the use of mid-ranks increases the evidence against the null hypothesis. Based on the randomization approach, we conclude that there is strong evidence against trivariate extreme-value dependence in the triples \{U,Co,Li\} and \{U,Li,Ti\}, while there is only marginal evidence against trivariate extreme-value dependence in the triple \{Ti,Li,Cs\}.

\subsection*{Acknowledgments}

The authors are very grateful to the associate editor and the referees for their constructive and insightful suggestions which helped to clean up a number of errors. The authors also thank Johanna Ne\v slehov\'{a} for providing R routines implementing the test based on $\hat \sigma_n^2$, and Mark Holmes for very fruitful discussions, as always.

\appendix

\section{Proofs of Propositions~1 and~2}
\label{proofs}

\begin{proof}[\bf Proof of Proposition~\ref{limitH0}]
From the limiting behavior of the empirical copula process given in Section~\ref{empirical_copula_process} and the functional version of Slutsky's theorem \citep[see e.g.,][Chap. 3.9]{vanWel96}, we have that
\begin{equation}
\label{limit}
\left(
\begin{array}{c}
\sqrt{n} \left [ \{C_n(\vec u^{1/r})\}^r - \{C(\vec u^{1/r})\}^r \right] \\
\sqrt{n} \left\{ C_n(\vec u) - C(\vec u) \right\} 
\end{array}
\right)
\leadsto 
\left(
\begin{array}{c}
r \{ C(\vec u^{1/r}) \}^{r-1}  \Cb(\vec u^{1/r}) \\
\Cb(\vec u) 
\end{array}
\right)
\end{equation}
in $\{\ell^\infty([0,1]^d)\}^2$. The desired result then follows from the continuous mapping theorem \citep[see e.g.][Theorem 1.3.6]{vanWel96}.
\end{proof}

\begin{proof}[\bf Proof of Proposition~\ref{multiplier}]
Let $j \in \{1,\dots,d\}$, and notice that
$$
C^{[j]}_n(\vec u) = \frac{1}{n(u_{j,n}^+ - u_{j,n}^-)}  \sum_{i=1}^n \left\{ \1(u_{j,n}^- < \hat U_{ij} \leq  u_{j,n}^+) \prod_{k=1 \atop k\neq j}^d \1(\hat U_{ik} \leq u_k) \right\}, \qquad \vec u \in [0,1]^d.
$$
Also, $1/\sqrt{n} \leq u_{j,n}^+ - u_{j,n}^- \leq 2/\sqrt{n}$ for all $u_j \in [0,1]$ and all $n \geq 1$. Hence,
$$
C^{[j]}_n(\vec u) \leq \frac{1}{\sqrt{n}}  \sum_{i=1}^n \1 \left\{ (n+1)u_{j,n}^- < R_{ij} \leq  (n+1)u_{j,n}^+ \right\}, \qquad \vec u \in [0,1]^d,
$$
where $R_{ij}$ is the rank of $X_{ij}$ among $X_{1j},\dots,X_{nj}$. It follows that
$$
\sup_{\vec u \in [0,1]} C^{[j]}_n(\vec u) \leq \sup_{\vec u \in [0,1]} \frac{(n+1)u_{j,n}^+ - (n+1)u_{j,n}^- + 1}{\sqrt{n}} \leq \frac{2(n+1)}{n} + \frac{1}{\sqrt{n}}.
$$
Thus, $C^{[j]}_n(\vec u) \leq 5$ for all $j \in \{1,\dots,d\}$, all $\vec u \in [0,1]^d$ and all $n \geq 1$. The latter fact combined with Lemma~\ref{uniform_convergence} given in Appendix~\ref{estimators_partial_deriv} stating a uniform convergence in probability of $C^{[j]}_n$ to $C^{[j]}$ enables us to use Proposition~3.2 of \cite{Seg11} from which it follows that
$$
\left(\Cb_n,\Cb_n^{(1)},\dots,\Cb_n^{(N)} \right) \leadsto \left( \Cb,\Cb^{(1)},\dots,\Cb^{(N)} \right)
$$ 
in $\{\ell^\infty([0,1]^d)\}^{(N+1)}$, where $\Cb^{(1)},\dots,\Cb^{(N)}$ are independent copies of $\Cb$. The desired result is then a consequence of the continuous mapping theorem and the fact that $C_n$ converges uniformly in probability to $C$.
\end{proof}

\section{Proof of Proposition~3}
\label{proof_Trn}

In order to prove the joint weak convergence of $T_{r,n},T_{r,n}^{(1)},\dots,T_{r,n}^{(N)}$, we first show a lemma.

Let $A$ be the space of bounded, Borel measurable functions on $[0, 1]^d$ and let $B$ be the space of c.d.f.s of finite Borel measures on $[0, 1]^d$. Define $\phi : A \times B \to \Rset$ by $\phi(a, b) = \int a \, \dd b$ and denote $\norm{f} = \sup_{u \in [0, 1]^d} |f(u)|$ for $f : [0, 1]^d \to \Rset$. 
The topologies on $A$ and $B$ are the ones induced by uniform convergence. The topology on $A \times B$,  $A^{N+1}$ or $A^{N+1} \times B$ is the product topology.

\begin{lem}
\label{lem:cont}
The map $\phi$ is continuous at each pair $(a_0, b_0)$ of $A \times B$ such that the functions $a_0$ and $b_0$ are continuous on $[0, 1]^d$.
\end{lem}

\begin{proof}
Let $(a_n, b_n)$ be a sequence in $A \times B$ such that $\norm{a_n - a_0} \to 0$ and $\norm{b_n - b_0} \to 0$. 
We have to show that $\int a_n \, \dd b_n \to \int a_0 \, \dd b_0$. 
Let $\beta_n$ and $\beta_0$ be the finite Borel measures on $[0, 1]^d$ associated with the c.d.f.s $b_n$ and $b_0$, respectively. By the triangle inequality,
$$
\left| \int a_n \, \dd b_n - \int a_0 \, \dd b_0 \right|
\le \left| \int a_n \, \dd b_n - \int a_0 \, \dd b_n \right|
+ \left| \int a_0 \, \dd b_n - \int a_0 \, \dd b_0 \right|.
$$
We treat the two terms on the right-hand side of the previous inequality separately.  

First, by uniform convergence and the continuity of $b_0$ on $[0,1]^d$, we have
$$
\beta_n([0, 1]^d) = b_n(1, \ldots, 1) \to b_0(1, \ldots, 1) = \beta_0([0, 1]^d) . 
$$
As a consequence,
$$
  \left| \int a_n \, \dd b_n - \int a_0 \, \dd b_n \right| \le \int | a_n - a_0 | \, \dd b_n \le \norm{a_n - a_0} \; \beta_n([0, 1]^d)
  \to 0. 
$$

Second, as $\norm{b_n - b_0} \to 0$, we 
have $\beta_n \to \beta_0$ in the topology of weak convergence of finite Borel measures. By continuity of the function $a_0$, this implies $\int a_0 \, \dd b_n \to \int a_0 \, \dd b_0$, 
as required.
\end{proof}

\begin{proof}[\bf Proof of Proposition~\ref{T_rn}]
The fact that $S_{r,n},S_{r,n}^{(1)},\dots,S_{r,n}^{(N)}$ jointly converge weakly to independent copies of the same limit is an immediate consequence of Proposition~\ref{multiplier} and the continuous mapping theorem. 

Let us show the corresponding result for $T_{r,n},T_{r,n}^{(1)},\dots,T_{r,n}^{(N)}$. Observe first that both $A$ and $B$, defined at the beginning of this Appendix, are subsets of the space $\ell^\infty([0, 1]^d)$. Next, the map $A^{N+1} \to A^{N+1} \times B : (a^{(1)},\dots,a^{(N+1)}) \mapsto (a^{(1)},\dots,a^{(N+1)},b)$ being continuous, we have, from Proposition~\ref{multiplier} and the continuous mapping theorem, that  
$$
\left(\Db_{r,n},\Db_{r,n}^{(1)}, \dots,\Db_{r,n}^{(N)}, C \right) \leadsto \left( \Db_r,\Db_r^{(1)},\dots,\Db_r^{(N)}, C \right).
$$ 
in $A^{N+1} \times B$. Since $\norm{C_n - C}$ converges to zero in probability, it follows that
\begin{equation}
\label{Drn+C}
\left(\Db_{r,n},\Db_{r,n}^{(1)}, \dots,\Db_{r,n}^{(N)}, C_n \right) \leadsto \left( \Db_r,\Db_r^{(1)},\dots,\Db_r^{(N)}, C \right).
\end{equation}
in $A^{N+1} \times B$.

Let $A_0 = \{ a \in A : \text{$a$ is continuous} \}$ and $B_0 = \{ b \in B : \text{$b$ is continuous} \}$. Copulas being continuous, $C$ belongs to $B_0$. From Proposition~\ref{limitH0}, we have that $\Db_r$ belongs to $A_0$ with probability one since the same is true for $\Cb$ defined in~(\ref{ec}). The limiting processes $\Db_r^{(1)},\dots,\Db_r^{(N)}$ also belong to $A_0$ with probability one since they are independent copies of $\Db_r$. 

By Lemma~\ref{lem:cont}, the map $\phi : A \times B \to \Rset$ is continuous at every point of $A_0 \times B_0$. It follows that the map $\psi : A^{N+1} \times B \to \Rset^{N+1}$ defined by
$$
\psi \left( a^{(1)},\dots,a^{(N+1)},b \right) = \left( \phi(a^{(1)},b),\dots,\phi(a^{(N+1)},b) \right)
$$
is continuous at every point of $A^{N+1}_0 \times B_0$. From~(\ref{Drn+C}) and the continuous mapping theorem, we finally obtain that
$$
\psi \left(\Db_{r,n},\Db_{r,n}^{(1)}, \dots,\Db_{r,n}^{(N)}, C_n \right) \leadsto \psi \left( \Db_r,\Db_r^{(1)},\dots,\Db_r^{(N)}, C \right),
$$
which is the desired result.
\end{proof}

\section{Estimators of the partial derivatives}
\label{estimators_partial_deriv}

In order to estimate the unknown partial derivative $C^{[j]}$, $j \in \{1,\dots,d\}$, besides the estimator $C^{[j]}_n$ defined in~(\ref{estimator}), one could use the estimator proposed by \citet{RemSca09}, and defined by  
\begin{multline*}
C^{[j]}_{n,RS}(\vec u) = \frac{1}{2 n^{-1/2}} \left\{ C_n (u _1,\dots,u_{j-1},u_{j,n}^+,u_{j+1},\dots,u_d ) \right. \\ \left. - C_n ( u_1,\dots,u_{j-1},u_{j,n}^-,u_{j+1},\dots,u_d ) \right\}, \qquad \vec u \in [0,1]^d,
\end{multline*}
where $u_{j,n}^+ = (u_j + n^{-1/2}) \wedge 1$, and $u_{j,n}^- = (u_j - n^{-1/2}) \vee  0$.

It is easy to verify that, for fixed $0 < a < b < 1$ and $n$ sufficiently large, $C^{[j]}_n$ and $C^{[j]}_{n,RS}$ coincide on $\{\vec u \in [0,1]^d : a \leq u_j \leq b \}$, and hence, from Lemma~\ref{uniform_convergence} below, if $C^{[j]}$ is continuous on the set $V_j$ defined in Condition ($\CC$), both estimators converge in probability to $C^{[j]}$ uniformly on $\{ \vec u \in [0,1]^d : a \leq u_j \leq b \}$. 

It could be argued that the following is a desirable property of an estimator of $C^{[j]}$: if $C^{[j]}$ happens to be continuous on $[0,1]^d$ instead of $V_j$, the estimator should converge in probability to $C^{[j]}$ uniformly on $[0,1]^d$. This property is satisfied by $C^{[j]}_n$ as is verified in Lemma~\ref{uniform_convergence} below. It is however not satisfied by $C^{[j]}_{n,RS}$ since the latter estimator does not converge pointwise in probability at points $\vec u$ of $[0,1]^d$ such that $u_j=0$ or $u_j=1$. 

\begin{lem}
\label{uniform_convergence}
Let $j \in \{1,\dots,d\}$, $0 < a < b < 1$, and assume that Condition ($\CC$) holds. Then, 
$$
 \sup_{\vec u \in [0,1]^d \atop u_j \in [a,b]} | C^{[j]}_n(\vec u) - C^{[j]}(\vec u)| \overset{\Pr}{\to} 0 .
$$
If, additionally, the partial derivative $C^{[j]}$ is continuous on $[0,1]^d$, then, 
$$
\sup_{\vec u \in [0,1]^d} | C^{[j]}_n(\vec u) - C^{[j]}(\vec u)| \overset{\Pr}{\to} 0 .
$$
\end{lem}

\begin{proof}
Without loss of generality, fix $j=1$, and, for any $\vec u \in [0,1]^d$, let $\vec u_{-1}$ denote the vector $(u_2,\dots,u_d)$ of $[0,1]^{d-1}$. Also, let $\delta > 0$ be a real number such that $0 < \delta < a < b < 1-\delta < 1$, and let $n$ be sufficiently large such that, for any $x \in [a,b]$, $x \pm n^{-1/2} \in [\delta,1-\delta]$. Now, for any $\vec u \in [0,1]^d$ such that $u_1 \in [a,b]$, we can write
\begin{align*}
C^{[1]}_n(\vec u) =& \frac{1}{u_{1,n}^+ - u_{1,n}^-} \left\{ C(u_{1,n}^+,\vec u_{-1}) - C(u_{1,n}^-,\vec u_{-1}) \right\} \\ &+ \frac{1}{u_{1,n}^+ - u_{1,n}^-} \left\{ (C_n - C)(u_{1,n}^+,\vec u_{-1}) - (C_n - C)(u_{1,n}^-,\vec u_{-1}) \right\}.
\end{align*}
From the fact that $u_{1,n}^+ - u_{1,n}^- = 2n^{-1/2}$ for all $u_1 \in [a,b]$, we obtain that
\begin{multline}
\label{decomp}
\sup_{\vec u \in [0,1]^d \atop u_1 \in [a,b]} | C^{[1]}_n(\vec u) - C^{[1]}(\vec u)| \leq \sup_{\vec u \in [0,1]^d \atop u_1 \in [a,b]} \left| \frac{1}{u_{1,n}^+ - u_{1,n}^-} \left\{ C(u_{1,n}^+,\vec u_{-1}) - C( u_{1,n}^-,\vec u_{-1}) \right\} - C^{[1]}(\vec u) \right| \\ 
+ \sup_{\vec u \in [0,1]^d \atop u_1 \in [a,b]} \left| \Cb_n(u_{1,n}^+,\vec u_{-1}) - \Cb_n( u_{1,n}^-,\vec u_{-1}) \right|,
\end{multline}
where $\Cb_n = \sqrt{n}(C_n - C)$. Since $C^{[1]}$ exists and is continuous on the set $V_1 = \{ \vec u \in [0,1]^d : 0 < u_1 < 1\}$, from the mean value theorem, we have that
$$
\frac{1}{u_{1,n}^+ - u_{1,n}^-}  \left\{ C(u_{1,n}^+,\vec u_{-1}) - C(u_{1,n}^-,\vec u_{-1}) \right\} = C^{[1]}(u_{1,n}^*,\vec u_{-1}),
$$
where $u_{1,n}^* \in (u_{1,n}^-,u_{1,n}^+) \subset [\delta,1-\delta]$. It follows that 
\begin{multline*}
\sup_{\vec u \in [0,1]^d \atop u_1 \in [a,b]} \left| \frac{1}{u_{1,n}^+ - u_{1,n}^-} \left\{ C(u_{1,n}^+,\vec u_{-1}) - C( u_{1,n}^-,\vec u_{-1}) \right\} - C^{[1]}(\vec u) \right| \\ = \sup_{\vec u \in [0,1]^d \atop u_1 \in [a,b]} \left| C^{[1]}(u_{1,n}^*,\vec u_{-1}) - C^{[1]}(\vec u) \right| \leq \sup_{(u',\vec u) \in [0,1]^{d+1} \atop { u',u_1 \in [\delta,1-\delta] \atop  |u' - u_1| \leq n^{-1/2}}} \left| C^{[1]}(u',\vec u_{-1}) - C^{[1]}(\vec u) \right|.
\end{multline*}
The term on the right converges to zero as $n$ tends to infinity because $C^{[1]}$ is uniformly continuous on the set $\{ \vec u \in [0,1]^d : \delta \leq u_1 \leq 1 - \delta \}$. 

The fact that
\begin{equation}
\label{term2}
\sup_{\vec u \in [0,1]^d \atop u_1 \in [a,b]} \left| \Cb_n(u_{1,n}^+,\vec u_{-1}) - \Cb_n(u_{1,n}^-,\vec u_{-1}) \right|  \overset{\Pr}{\to} 0
\end{equation}
is a consequence of the asymptotic equicontinuity of the sequence $\Cb_n$, which follows from the weak convergence of $\Cb_n$ in $\ell^\infty([0,1]^d)$ to the Gaussian process with continuous paths $\Cb$ defined in~(\ref{ev}) \citep[see e.g.,][page 115 and Equation~(2.6)]{Kos08}.

To prove the second statement, notice that~(\ref{decomp}) with $a=0$ and $b=1$ holds for all $n \geq 1$ because $0 < 1/(u_{1,n}^+ - u_{1,n}^-) \leq \sqrt{n}$ for all $u_1 \in [0,1]$ and all $n \geq 1$. Then, applying the mean value theorem as previously, one obtains that 
\begin{multline*}
\sup_{\vec u \in [0,1]^d} \left| \frac{1}{u_{1,n}^+ - u_{1,n}^-} \left\{ C(u_{1,n}^+,\vec u_{-1}) - C( u_{1,n}^-,\vec u_{-1}) \right\} - C^{[1]}(\vec u) \right| \\ = \sup_{\vec u \in [0,1]^d} \left| C^{[1]}(u_{1,n}^*,\vec u_{-1}) - C^{[1]}(\vec u) \right| \leq \sup_{(u',\vec u) \in [0,1]^{d+1} \atop |u' - u_1| \leq n^{-1/2}} \left| C^{[1]}(u',\vec u_{-1}) - C^{[1]}(\vec u) \right|,
\end{multline*}
where $u_{1,n}^* \in (u_{1,n}^-,u_{1,n}^+) \subset [0,1]$. The term on the right of the previous display converges to zero as $n$ tends to infinity because $C^{[1]}$ is uniformly continuous on $[0,1]^d$. The desired results finally follows from the fact that the result stated in~(\ref{term2}) with $a=0$ and $b=1$ holds for the same reasons as previously.
\end{proof}

\small
\bibliographystyle{plainnat}
\bibliography{biblio}

\end{document}